\newif\ifcomments  
\commentsfalse
\documentclass{article}
\usepackage{amsmath,amssymb,pifont}
\usepackage{tcolorbox}
\usepackage{multicol}
\usepackage{amstext}
\usepackage{amsthm}
\usepackage{multirow}
\usepackage{booktabs}
\usepackage[skip=0pt]{subcaption}
\usepackage{times}
\usepackage{lipsum}
\usepackage[shortlabels]{enumitem}
\usepackage{cancel}
\usepackage{wrapfig}
\usepackage{array}
\usepackage{siunitx}
\usepackage{csvsimple}
\usepackage[multidot]{grffile}
\usepackage{bbm}
\usepackage{dblfloatfix}
\usepackage{hyperref}
\usepackage{makecell}
\usepackage{bbm, dsfont}
\usepackage{mathtools}
\usepackage{comment}
\usepackage[capitalise]{cleveref}

\usepackage{geometry}
\usepackage{enumitem}

\usepackage[multiple]{footmisc}
\usepackage{mathrsfs}
\usepackage{todonotes}
\usepackage{tikz}
\usepackage{hyperref}
\usepackage{cleveref}
\usepackage{algpseudocode,algorithm,algorithmicx}
\usepackage{xfrac}

\usepackage{graphicx} 
\usepackage{color}

\usepackage{array}
\usepackage{amssymb}
\usepackage{amsmath}
\usepackage{xspace}
\usepackage{fancyhdr}
\usepackage{comment}

\hypersetup{final}

\newcommand{\bfg}{\ensuremath{\mathbf{g}}}

\newcommand{\bfv}{\ensuremath{\mathbf{v}}}

\newcommand{\bfz}{\ensuremath{\mathbf{z}}}

\newcommand{\calB}{\ensuremath{\mathcal{B}}}

\newcommand{\calD}{\ensuremath{\mathcal{D}}}

\newcommand{\calF}{\ensuremath{\mathcal{F}}}

\newcommand{\calP}{\ensuremath{\mathcal{P}}}

\renewcommand{\Pr}{\mathop{\mathbf{Pr}}}

\newtheorem{lemma}{Lemma}[section]

\newtheorem{definition}[lemma]{Definition}

\makeatletter
\newcommand{\vast}{\bBigg@{4}}
\newcommand{\Vast}{\bBigg@{5}}
\makeatother

\newcommand{\ex}[2]{{\ifx&#1& \mathbb{E} \else
\underset{#1}{\mathbb{E}} \fi \left[#2\right]}}
\newcommand{\pr}[2]{{\ifx&#1& \mathbb{P} \else
\underset{#1}{\mathbb{P}} \fi \left[#2\right]}}

\newcommand{\ltwo}[1]{\left\|#1\right\|_2}

\usepackage{ulem}

\DeclarePairedDelimiterX{\infdivx}[2]{(}{)}{%
  #1\;\delimsize\|\;#2%
}

\newcommand{\mypar}[1]{\smallskip
	\noindent{\textbf{{#1}:}}}
	
\renewcommand{\epsilon}{\varepsilon}

\newcommand{\clip}[2]{{\sf clip}\left(#1,#2\right)}

\setlist{nolistsep}
\setlist[itemize]{noitemsep, topsep=0pt}

\setlist{nolistsep}
\setlist[itemize]{noitemsep, topsep=0pt}

\newenvironment{AlgorithmBox}[1][Algorithm]
{%
  \begin{tcolorbox}[
    title=#1,         
    colback=white,    
    colframe=blue,   
    boxsep=5pt,       
    arc=3pt,          
    outer arc=3pt,    
    fonttitle=\bfseries
  ]
}
{%
  \end{tcolorbox}
}
\definecolor{darkgreen}{RGB}{46,139,87}
\usepackage{fullpage}

\usepackage[numbers, sort, comma, square]{natbib}

\title{Hush! Protecting Secrets During Model Training:\\
An Indistinguishability Approach}
\author{
Arun Ganesh\thanks{Google Research \texttt{arunganesh@google.com}} \and
Brendan McMahan\thanks{Google Research \texttt{mcmahan@google.com}} \and
Milad Nasr\thanks{Google DeepMind. \texttt{srxzr@google.com}} \and
Thomas Steinke\thanks{Google DeepMind \texttt{steinke@google.com}}\and
Abhradeep Thakurta\thanks{Google DeepMind \texttt{athakurta@google.com}}
}

\begin{document}

\maketitle

\begin{abstract}
We consider the problem of \textit{secret protection}, in which a business or organization wishes to train a model on their own data, while attempting to not leak secrets potentially contained in that data via the model. The standard method for training models to avoid memorization of secret information is via differential privacy (DP). However, DP requires a large loss in utility or a large dataset to achieve its strict privacy definition, which may be unnecessary in our setting where the data curator and data owner are the same entity. We propose an alternate definition of secret protection that instead of targeting DP, instead targets a bound on the posterior probability of secret reconstruction. We then propose and empirically evaluate an algorithm for model training with this secret protection definition. Our algorithm solves a linear program to assign weights to examples based on the desired per-secret protections, and then performs Poisson sampling using these weights. We show our algorithm significantly outperforms the baseline of running DP-SGD on the whole dataset.
\end{abstract}

\section{Introduction}
\label{sec:intro}

We consider the problem of an organization training a model on an internal dataset containing some number of secrets. For example, consider an algorithmic trading company. They company could be interested in fine-tuning an internal code completion model used by all employees on their codebase because they are using their own internal libraries that a public code completion model would have no knowledge of. However, their code base likely contains information that they would not like their entire employee base to have access to, and hence want to prevent the model from memorizing. For example, they may have many teams working on different trading algorithms for different markets, and keep different teams' codebases and knowledge siloed to reduce the risk of leaking any of their trade secrets which give them a competitive advantage. In this case, if they trained the code completion model without any safeguards, it's possible the model could be prompted by any employee to reproduce code they should not be able to access. While the definitions and techniques in this paper could be applied broadly beyond this example, to keep the discussion concrete we will focus on this example throughout the paper.

A standard technique for preventing memorization of sensitive information is differential privacy (DP) \cite{DMNS}, and we could use an algorithm like DP-SGD \cite{song2013stochastic, BST14} or its variants, which clip gradients to a bounded $\ell_2$-norm and then add noise to prevent memorization of specific examples. To protect the secrets in a dataset, we could target e.g. 'secret-level' DP, i.e. ensure that the outputs of training on the dataset and a copy of the dataset with each example containing a certain secret being changed arbitrarily satisfy some indistinguishability guarantee.  However, DP gives strong protections necessary for a setting where a data curator is using external data belonging to third-party users. The example trading company is both the model trainer and the owner of the codebase they are training on, hence these strong protections are unnecessary. In this paper, we propose alternatives to DP for defining formal secret protection while training models, propose methods for achieving these definitions, and give empirical evaluations of the practicality of these goals methods.

\subsection{Our Contributions}

Our main contributions are a definition of secret protection that is more suitable for the threat model we study, and algorithms tailored to this definition.

\mypar{A definition for secret protection} In Section~\ref{sec:problemForumuation} we give our proposed definition of secret protection. Our definition differs from standard DP guarantees in the following ways: (i) It gives each secret a different degree of protection, rather than requiring a uniform protection (ii) It assumes the existence of secrets is public information, rather than trying to keep membership private (iii) It does not allow examples to arbitrarily change for two adjacent datasets.

Our definition bounds the posterior reconstruction probability for each secret given a prior reconstruction probability. While tools from \cite{hayes2023bounding} do this tightly if we can exactly compute the privacy loss distribution of our algorithm, in our empirical setting they suffer from numerical instability issues. We found divergence bounds to be more robust to these issues, and give a divergence-based method for bounding reconstruction probabilities (\cref{lem:generalized_fano}) that is tight for a given divergence bound. This improves on the loose R\'enyi-divergence-based reconstruction bounds of \cite{balle22reconstructing}.

\mypar{Algorithms tailored to our new definition} We next study algorithms for model training with secret protection. The main algorithmic challenge is that given a dataset, even if each secret has a uniform privacy guarantee we are targeting, some secrets may appear more frequently in the dataset and require more noise. Instead, we would like to bound the number of appearances of each secret in the dataset, so that each secret appears a number of times roughly proportional to its desired privacy guarantee. We propose a simple but novel algorithm for selecting a subset of the dataset to use in training that achieves this property. At a high level, our algorithm uses a linear program to assign each example a fractional weight, and uses these weights to choose sampling probabilities for each example when forming batches. We empirically demonstrate that our algorithm allows us to achieve a $\approx 8 \times$ reduction in the noise and $\approx 7.5\%$ reduction in test loss compared to not doing any preprocessing of the dataset. 

\subsection{Related Work}\label{sec:relatedwork}

The observation that DP even with weak privacy parameters protects against reconstruction and suggestion to relax the DP parameters when membership inference is not a concern are well-known. \cite{bhowmick19protection}, motivated by this observation, give minimax optimal local $\epsilon$-DP mechanisms in the high-$\epsilon$ regime for noising vectors, as opposed to past work which only considered the low-$\epsilon$ regime. \cite{hayes2023bounding} give upper bounds on the probability of reconstructing an example included in DP-SGD, and an attack which gives a nearly matching empirical lower bound. They also show different algorithms with the same approximate DP parameters can give very different protections against reconstruction, showing standard DP definitions are not good predictors of protection against reconstruction. As we will discuss later in \cref{sec:protReconst}, the bounds of \cite{hayes2023bounding} are tight but can require highly numerically unstable computations in some cases.

A key component of our definition of secret protection is that we offer different privacy protections to different secrets, whereas most of the DP literature studies uniform privacy protections. \cite{yu2023individual} observed that DP-SGD often provides different privacy protections to different individuals. However, this is an artifact of average-case training behaviors that occur in practice, and in their work DP-SGD still is implemented to target a uniform worst-case privacy guarantee, as opposed to our algorithms which explicitly target non-uniform worst-case protections. 

Our definition of secret protection is somewhat analogous to user-level DP, where the privacy unit is all examples owned by a single user. How to fit a standard algorithm like DP-SGD to user-level DP is well-studied in the setting where each example is owned by a single user \cite{charles2024finetuninglargelanguagemodels, chua2024mind}. Across the broader DP literature, the setting where an example can be owned by multiple users has been indirectly studied in algorithms for queries on graphs. For example, \cite{dong2022r2t} study the problem of graph pattern counting under node-level DP. Here, each instance of a graph pattern can be viewed as an example belonging to all the nodes in it. \cite{ganesh2025itsmydatatoo} recently proposed and studied a \textit{multi-attribution} DP definition, similar to node-level DP but where the graph is assumed to be public, and with a focus on model training rather than graph queries.

\section{Background on DP-SGD and Privacy Accounting}

DP-SGD \cite{song2013stochastic, BST14, DP-DL}, formally stated as \cref{fig:dpsgd}, is a canonical private learning algorithm. In each round it samples a batch of examples, computes the gradient of the loss on these examples, clips the gradients to a bounded $\ell_2$-norm, and adds Gaussian noise to privatize their sum. The privatized sum is then passed to a first-order method such as SGD or Adam.

\begin{figure}
\begin{algorithm}[H]
\caption{DP-SGD with Poisson sampling}
\label{fig:dpsgd}
\textbf{Inputs:} Dataset $D = \{x_1, \ldots, x_n\}$, loss $\ell$, optimizer update function $update$, secrets $S$, number of rounds $T$, target batch size $B$, clip norm $C$, noise multiplier $\sigma$, initial model $\theta_0$.
\begin{algorithmic}[1]
\State $\clip{\bfv}{C} := \bfv \cdot \min\{1, C / \ltwo{\bfv}\}$
\For{$t \in [T]$}
\State $\calB_i \leftarrow$ batch including $x_i$ independently w.p. $p$.
\State $\bfg_i \leftarrow \frac{1}{B}\sum_{x_i \in \calB_i} \nabla \clip{\ell(\theta_{i-1}; x_i)}{C}$
\State $\bfz_i \leftarrow N(0, \frac{C^2 \sigma^2}{B^2} \mathbb{I})$
\State $\theta_i \leftarrow update(\theta_{i-i}, \bfg_i + \bfz_i)$. \Comment{e.g. for SGD, $update(\theta, \bfg) = \theta - \eta \bfg$}
\EndFor
\end{algorithmic}
\end{algorithm}
\end{figure}

DP-SGD is commonly analyzed via \textit{privacy loss distributions}. Given two distributions $P, Q$, their \textit{privacy loss random variable} is given by $\log(P(x) / Q(x)), x \sim P$. The privacy loss distribution is the distribution of this random variable.
The \textit{blow-up function} is closely related to the privacy loss distribution. Given two distributions $P, Q$, their blow-up function is the function $B_{P,Q}(x) := \max_{E: Q(E) \leq x} P(E)$. In other words, the blow-up function gives the maximum true positive rate of a binary classifier to identify samples from $Q$, with false positive rate at most $x$ on $P$. \cite{zhu22characteristic} observed that the privacy loss distribution and blow-up function are equivalent representations, i.e. the privacy loss distribution of $P, Q$ is easily computable from its blow-up function and vice-versa. 

Under a given definition of adjacency $\simeq$ (i.e., a relation that defines what pairs of datasets are adjacent), a \textit{dominating pair} $P, Q$ for a mechanism $M$ is a pair such that $B_{P, Q}(x) \geq B_{M(D), M(D')}(x)$ for all $x \in [0, 1], D \simeq D'$. A dominating pair $P, Q$ is \textit{tight} if there exists $D \simeq D'$ such that $B_{P, Q} = B_{M(D), M(D')}$. For DP-SGD when $D \simeq D'$ if $D, D'$ differ in a group of examples of bounded size, tight dominating pairs are well-known \cite{charles2024finetuninglargelanguagemodels} and their privacy loss distributions are easily computable via open-source software \cite{dp_accounting}.

\section{Secret Protection}
\label{sec:problemForumuation}

We now introduce our proposed definition of secret protection:

\begin{definition}
Suppose as input, an algorithm $A$ receives
\begin{itemize}[leftmargin=1em]
    \item A dataset $D = \{x_1, \ldots, x_n\} \in \calD^n$ in data universe $\calD$,
    \item A list of secrets $S = \{y_1, \ldots, y_m\}$, and for each secret $y_j$ a predicate $T_j: \calD^2 \rightarrow \{0, 1\}$, where $T_j(x_i, x_i') = 1$ means that example $x_i$ is equivalent to $x_i'$ with only the secret $y_j$ transformed.
    \item A function $E: D \rightarrow 2^S$, where $E(x_i)$ represents all secrets in $S$ contained in $x_i$.
\end{itemize}

We define inputs $I = (D = \{x_1, \ldots, x_n\}, S, E)$, $I' = (D' = \{x_1', \ldots, x_n'\}, S, E)$, as differing only in secret $j$, denoted $I \simeq_j I'$, if for all $i \in [n]$ either $y_j \notin E(x_i)$ or $T_j(x_i, x_i') = 1$. In other words, $D$ and $D'$ differ only in examples that contain the secret $y_j$, where they only differ by information pertaining to the secret, and otherwise the inputs (including $S$ and $E$) don't differ. 

Then, $A$ satisfies $ \{(p_j, r_j)\}$-secret protection if it satisfies the following: For any $j$, given a distribution $\calP$ over $[k]$ such that $\calP(i) \leq p_j$ for all $i$ and a set of candidate inputs $\{I_1, I_2, \ldots, I_k\}$ such that for any $i, i'$, $I_{i} \simeq_j I_{i'}$, for any adversary $B$, $Pr_{i \sim \calP, A}[B(A(I_i)) = i] \leq r_j$. 

\label{def:secretprotection}
\end{definition}

Suppose each all $p_j$ were equal and all $r_j$ were equal, $T_j$ was always equal to 1, and we enforced the bound of $r_j$ via a divergence (such as a R\'enyi or hockey-stick divergence) bound on $A(I), A(I')$ for all input pairs $I, I': \exists j: I \simeq_j I'$. If we treated each secret $y_j$ as a user, and $D_j$ as the set of examples owned by that user \cref{def:secretprotection} would recover a DP definition for user-level DP, where an example can be owned by multiple users, and which examples are owned by which users is public information.
Hence, the key differences between \cref{def:secretprotection} and user-level DP are as follows:

\begin{itemize}[leftmargin=1em]
    \item \textbf{Public membership:} We do not protect against knowing about the existence of secrets and which examples contain which secrets, just the contents of examples (i.e., the secrets themselves). This is encoded in the definition by the fact that the structure of the secrets given by $\{D_1, \dots, D_m\}$ is shared between adjacent inputs. We note that this is analogous to the fixed-graph DP definition recently proposed in \cite{ganesh2025itsmydatatoo}. 
    \item \textbf{Non-uniform protection:} We do not require uniform ``privacy parameters'' for the secrets, as not all secrets are equally sensitive.
    \item \textbf{Privacy guarantee:} Our DP definition requires us to protect against reconstruction attacks, rather than bound the divergence of all pairs of outputs given by adjacent inputs. A divergence bound usually implies a reconstruction bound as we will demonstrate in \cref{lem:generalized_fano}, so this can be viewed as a loosening of the privacy guarantee.
    \item \textbf{Restricted adjacency:} The restriction of adjacency via $T_j$ allows us to release public information within examples that also contain secrets without violating the privacy definition.
\end{itemize}

\subsection{Why Is DP Overambitious?}\label{sec:overambitious}

We identify reasons why DP may be overly protective for the needs of secret protection, and state weaker goals for model training encoded in \cref{def:secretprotection} that allow us to improve utility. 

First, as previously mentioned, differential privacy often targets a setting where the model trainer (e.g., a data curator working for a hospital) is a different entity than the data owners with privacy concerns (e.g., the hospital's patients). In turn, strong privacy protection is usually an ethical and/or legal requirement for using the data. In contrast, we are focused on a setting where \textit{the model trainer and data owner are the same entity}, and hence they may be willing to forgo the strong protection of DP in favor of utility. For example, a trading company may be willing to tolerate lower privacy standards for training on its own codebase, if they believe the benefits of having a well-trained model outweigh the losses of forgoing very strong privacy guarantees.

Second, differential privacy often provides a uniform privacy protection to all involved users (with some exceptions we discuss later). However, \textit{not all secrets are equally sensitive}. By blindly applying the standard definitions of DP (which provide a uniform privacy protection to all secrets) the company might sacrifice utility to protect information that the company is willing to risk a higher chance of leaking. For example, in the trading company's codebase, consider the following two Python files:

\begin{figure}[H]
\begin{minipage}{0.5\textwidth}
\begin{AlgorithmBox}[Example of \texttt{secret\_investment.py}]
{\color{red}
\begin{verbatim}
"""Trains an investment model
for a secret purpose."""

import file_utils
import torch

def training_loop():
    ...
\end{verbatim}
}
\end{AlgorithmBox}
    
\end{minipage}
\begin{minipage}{0.5\textwidth}
\begin{AlgorithmBox}[Example of \texttt{file\_utils.py}]
{\color{darkgreen}
\begin{verbatim}
"""Utils for accessing file
system during model
training."""

def load_dataset(name: str):
    ...
    

\end{verbatim}
}
\end{AlgorithmBox}
\end{minipage}
\end{figure}

\texttt{secret\_investment.py} is a file which contains details of how a trading model is defined and trained for an investment the trading company wishes to keep secret, and hence is likely highly sensitive. \texttt{file\_utils.py} contains tools for accessing the file system, which may be still be mildly sensitive but are less likely to be damaging if leaked. So we might tolerate weaker privacy protections for \texttt{file\_utils.py} than for \texttt{secret\_investment.py}. 
We note that non-uniform privacy protections have appeared in various forms in the literature, and we do not claim to be the first to propose this relaxation of privacy. For example \cite{jorgensen15pdp} proposed ``personalized'' DP, where each user specifies their own privacy budget. However, an important qualitative difference between our secret protection definition and these definitions is the assumptions required for a non-uniform privacy protection to offer improvements in utility. When we have data owners who are a different entity than the model trainer and have no investment in the success of the model trainer, data owners who act in their own self-interest will always opt for the maximum privacy budget available to them. Furthermore, legal obligations such as GDPR compliance may obligate model trainers to achieve a target DP guarantee for all users \cite{cummings18therole}, even if some users may be unconcerned with their privacy. Hence in these settings, relaxations such as personalized DP may be inapplicable, or may not offer any improvements in utility when they are applicable. In contrast, in our setting the model trainer and data owner are the same entity, and hence the legal requirement does not exist, and loosening privacy of some secrets may be in the data owner's own best interest. So, the gains in utility from allowing non-uniform protections are more likely to be realizable.

Finally, differential privacy usually protects against leaking any bit of information from an example, whereas we are only concerned with protecting the secrets, which are typically larger than one bit. This allows us to relax the privacy definition in several different ways. First, while DP protects against membership inference (i.e., knowing whether or not an example is in a dataset), the mere existence of a business secret may not be a privacy concern. For example, an employee might be able to see that there is a \texttt{secret\_investment} library, but without knowing more details this is unproblematic. In this setting, our privacy definition does not need to protect against membership inference, and we can tailor the definition to protecting against stronger privacy violations like reconstructing the contents of  \texttt{secret\_investment.py}. Second, even for the problem of reconstruction, not all parts of a training example may be problematic to reconstruct. For example, if \texttt{secret\_investment.py} is used as an example in training, the fact that \texttt{file\_utils} and \texttt{torch} are imported together is fine to release even if the learning algorithm in the same file is a secret.

\subsection{Defining $S, E, T$}

In our secret protection definition we have deliberately left the definition of $S, E, T_j$, i.e. the secrets, which examples include which secrets, and what constitutes changing only a secret in an example, to be specified by the model trainer. We emphasize that while we have treated these as givens, defining these quantities is perhaps the most challenging part of using our secret protection definition in practice. The main challenge is that in contrast with e.g. user-level DP, where an example may come with metadata (e.g., the username of the account that uploaded the example) that easily lets us identify users and associations between examples and users, a definition of a secret and identifying an example as containing a secret are both likely more subjective.

We believe designing effective techniques for identifying secrets and examples containing secrets in various types of datasets is an important open problem. For our specific motivating example of a code dataset, we believe there may be hope in automating the process of identifying secrets and examples containing them, using structure built into the dataset. Dependencies between code libraries can hint at where in the codebase a secret may reappear. For example, if an algorithm which is a trade secret appears in a library that is commonly imported, it may be reasonable to tag other libraries importing this library as containing the secret. Similarly, code is often siloed by user groups that could hint at secret containment, and also has a commit history that associates code with certain contributors. By identifying user groups and code contributors working on sensitive projects, we could infer that files which they have modified contain a secret associated with the sensitive project. 

Unfortunately, because datasets containing actual secrets cannot be made public or used for public research, we did not pursue an empirical investigation of techniques for defining these quantities as part of this work. We hope that future work can make progress on this  problem.

\subsection{Bounding the Posterior Probability}\label{sec:protReconst}

We now discuss different strategies for bounding the posterior reconstruction probability $r_j$ given the prior probability $p_j$. Recall that the \textit{blow-up function} of two distributions $P, Q$ is $B_{P, Q}(x) := \max_{E: Q(E) \leq x} P(E).$ Theorem 2 of~\cite{hayes2023bounding} shows that if for all $j$ and any inputs $I \simeq_j I'$ that $B_{A(I), A(I')}(p_j) \leq r_j$, this suffices for $\{(p_j, r_j)\}$-secret protection, and this is tight, i.e. one cannot achieve smaller $r_j$ than $\max_{I \simeq I'} B_{A(I), A(I')}(p_j)$. Also recall that for DP-SGD, a ``worst-case'' privacy loss distribution and hence blow-up function are known, even for group privacy. Hence, in principle we can use this worst-case blow-up function to compute tight $r_j$ values for DP-SGD.

However, a practical issue with working with privacy loss distributions is that the distribution is continuous, and hence we need to (i) discretize the distribution at some cost in accuracy, and (ii) deal with numerical errors due to working with floating-point precision. This can be handled in a pessimistic manner, i.e. any errors will only result in over-estimating parameters that measure privacy leakage such as the posterior reconstruction probabilities $r_j$, so the reported parameters are safe to use in practice. However, as we will discuss in further detail in \cref{sec:expsecretprotection}, in some settings this makes the computation $r_j$ via blow-up functions extremely unstable, to the point where the over-estimation is often too large to be usable. However, we found that the KL divergence, which is the mean of the privacy loss distribution, remained relatively stable. We can use the following lemma to bound posterior reconstruction probabilities via KL divergence and other divergences:

\begin{lemma}\label{lem:generalized_fano}
Let $\calP$ be a prior distribution over $[k]$ such that $\calP(i) \leq p$ for all $i \in [k]$, let $\calF$ be an $f$-divergence function for strictly convex $f$, i.e. $\calF(P, Q) := \mathbb{E}_{x \sim Q}\left[f(P(x) / Q(x))\right]$ for some convex $f$ such that $f(1) = 0$. Let $P_1, \ldots, P_k$ be distributions such that for all $i, j$ pairs $\calF(P_i, P_j) \leq \calF(Bern(q), Bern(p))$ where $q \geq p$. Then for any adversary $B$, $\Pr_{i \sim \calP, y \sim P_i}[B(y) = i] \leq q$.
\end{lemma}
\begin{proof}
This is a generalization of Fano's inequality, and we prove it similarly. Consider the sampling process $i \sim \calP, y \sim P_i, i' \sim \calP, y' \sim P_{i'}$. By (1) the post-processing property of $f$-divergences $\calF(P, Q) \geq \calF(g(P), g(Q))$ and (2) convexity of $f$-divergences: for $w_i \in (0, 1)$, $\calF(\sum_i w_i P_i, \sum_i w_i Q_i) < \sum_i w_i \calF(P_i, Q_i)$, we have (abusing notation to let a random variable denote a distribution):

\begin{align*}
\calF(\mathds{1}(B(y) = i), \mathds{1}(B(y') = i)) &\stackrel{(1)}{\leq} \calF((y, i), (y', i)) \\
&= \calF(\sum_{i, i'} \calP(i) \calP(i') (P_i, i), \sum_{i, i'} \calP(i) \calP(i') (P_{i'}, i))\\ &\stackrel{(2)}{\leq} \max_{i, i'} \calF(P_i, P_{i'}).\\
&\leq \calF(Bern(q), Bern(p)).
\end{align*}

Now assume that $Pr[B(y) = i] > q$. Since $f$ is strictly convex, this implies that for $q' > q \geq p \geq p'$, $\calF(Bern(q'), Bern(p')) > \calF(Bern(q), Bern(p))$.  For any $B$, we have $Pr[B(y') = i] \leq p$, since $B$ has no knowledge of $i$. So if $Pr[B(y) = i] > q$, we have $Pr[B(y) = i] > q \geq p \geq Pr[B(y') = i]$ and hence $\calF(\mathds{1}(B(y) = i), \mathds{1}(B(y') = i)) > \calF(Bern(q), Bern(p))$, which is a contradiction.
\end{proof}

The KL divergence is the $f$-divergence given by $f(y) = y \log y$, which is strictly convex. So another strategy for bounding the posterior reconstruction probability of $y_j$, $r_j$, given a worst-case privacy loss distribution for $A(I), A(I'), I \simeq_j I'$ is to compute the mean of the privacy loss distribution $\mu$, then compute the $r_j > p_j$ such that $KL(Bern(r_j), Bern(p_j)) = \mu$. If $\mu$ is computed from an algorithm with a noise level, we can use binary search to invert this procedure, finding the minimum noise level that achieves KL divergence bound $KL(Bern(r_j), Bern(p_j))$ for our target $r_j$. Unlike using the blow-up function to compute $r_j$, this is not tight (and cannot be, as the blow-up function $T:[0, 1] \rightarrow [0, 1]$ contains strictly more information than the single parameter $\mu$), but in our experiments proved to be more numerically stable.

\section{Algorithms for Secret Protection}\label{sec:algorithms}

Given \cref{def:secretprotection}, we can now consider algorithms for model training that achieve secret protection. As discussed in Section~\ref{sec:protReconst}, it is straightforward to derive posterior reconstruction probability bounds for DP-SGD. However, if we train on the entire dataset with DP-SGD, some secrets may appear many more times than others (relative to their desired level of secret protection), and for these secrets we may unnecessarily increase the noise multiplier in DP-SGD, worsening its utility.  Hence to achieve $\{(p_j, r_j)\}$-secret protection at minimal cost in utility, we want to preprocess the dataset such that each secret $y_j$ appears in the dataset a number of times roughly equal to an appropriate function of its corresponding $(p_j, r_j)$. At the same time, we want to make sure that when bounding the number of appearances of each secret, we do not discard too much of the data.

Our main algorithmic idea is to formulate a linear program for the problem of selecting a subset of the data such that we bound the number of appearances of each secret while keeping as many examples as possible. The linear program solution gives weights for each example, and then we can use Poisson sampling to form batches where each example's sampling probability is proportional to its assigned weight. That is, we will first compute a weight $w_i$ for each example $i$. Then, if we sample example $i$ with probability $\frac{B w_i}{ \sum_{i'} w_{i'}}$ (which is in $[0, 1]$ assuming $\max_i w_i \leq \sum_{i'} w_{i'} / B$), our expected batch size is $B$. Once we've decided on the sampling probabilities, we can set the noise multiplier to the minimum value such that each secret's target privacy guarantee is achieved.

We formulate the following linear program for assigning weights to examples:

\begin{equation}
\max  \sum_{i \in D} w_i, \qquad
s.t. \forall j \in S: \sum_{i: (i, j) \in E} w_i \leq c \mu_j, \qquad
\forall i \in D: w_i \in [0, 1].\label{eq:secret_lp}
\end{equation}

The objective is to maximize the sum of the weights, i.e. include as many examples as possible. We constrain $w_i$ to be in $[0, 1]$, which ensures that $\frac{B w_i}{ \sum_{i'} w_{i'}}$ is a probability as long as we find a solution to the LP where $\sum_{i'} w_{i'} > B$. For each secret $j$, we write a constraint saying the sum of the weights assigned to it is at most a target $\mu_j$, times a constant $c$ which we set before solving the program.

We set $\mu_j$ to be the KL divergence that guarantees posterior reconstruction probability at most $r_j$ by \cref{lem:generalized_fano}. This is based on the observation by past work that DP-SGD's group privacy parameter is roughly linear in the group size \cite{charles2024finetuninglargelanguagemodels}, which $\sum_{i: (i, j) \in E} w_i$ is a fractional version of. 
We note that ideally, we would include the noise multiplier $\sigma$ as a variable in the program and encode the downstream accounting that computes $\mu_j$ as a constraint. However, this accounting is computationally expensive and potentially would give a non-convex constraint.

The choice of the constant $c$ serves to trade off between preventing discarding too many examples and avoiding overrepresented secrets. When $c$ is sufficiently large, the constraints $\sum_{i: (i, j) \in E} w_i \leq c \mu_j$ can never be saturated, i.e. only the constraint $w_i \in [0, 1]$ matters and hence the optimal solution is just to set all $w_i = 1$. This is equivalent to running DP-SGD on the full dataset, which includes as many examples as possible but will require a high noise multiplier due to some secrets being overrepresented. On the other hand, when $c$ is sufficiently small, the constraint $w_i \leq 1$ can never be saturated because the constraints $\sum_{i: (i, j) \in E} w_i \leq c \mu_j$ are stricter. By well-known properties of basic feasible solutions to linear programs, $n = |D|$ constraints must be tight for the optimal solution, and there are only $m = |S|$ constraints of the form $\sum_{i: (i, j) \in E} w_i \leq c \mu_j$, which means $n - m$ constraints of the form $w_i \geq 0$ will be tight. In other words, if $c$ is too small, the solution will assign $w_i = 0$ to all but $m$ examples, i.e. discard a large fraction of examples. As we interpolate between large and small values of $c$, we will trade off between these two extreme solutions. In our experiments, we include $c$ as a hyperparameter to be tuned. Note that under our adjacency definition, the linear program (and its solution) can be treated as public information since we assume $G$ is public. Hence, sweeping $c$ has the same privacy ramifications as sweeping other parameters like learning rate, whose privacy cost is commonly ignored both in research settings and in practice.

\mypar{Accounting for secret protection}
A slight variation of Theorem 1 of \cite{charles2024finetuninglargelanguagemodels} shows that if examples in $E^{-1}(y_j) := \{x \in D: y_j \in E(x)\}$ have sampling probabilities $\rho_1, \ldots, \rho_k$ and we do $T$ rounds of DP-SGD with Poisson sampling, then $N(\sum_{i=1}^k Bern(\rho_i), \sigma^2)^T$, $N(0, \sigma^2)^T$ is a dominating pair for DP-SGD with noise multiplier $\sigma$ and pairs of inputs differing only in $E^{-1}(y_j)$. The open source \texttt{dp\_accounting} library supports computing the privacy loss distribution for this pair of distributions, and hence the KL divergence between the dominating pair which is an upper bound on the KL divergence between the outputs of DP-SGD. Hence using \cref{lem:generalized_fano}, given the sampling probabilities $\rho_j$ from the linear program, we can easily compute an indistinguishability guarantee for each secret. Then given the sampling probabilities, we can perform binary search on $\sigma$ to find (approximately) the minimum $\sigma$ such that all target indistinguishability guarantees are met, and run DP-SGD with these sampling probabilities and noise multiplier. 
We summarize our end-to-end training procedure, including the noise calibration, in \cref{fig:secret_protection_algorithm}.

\begin{figure}
\begin{algorithm}[H]
\caption{Secret protection training pipeline}
\label{fig:secret_protection_algorithm}
\textbf{Inputs:} Dataset $D = \{x_1, \ldots, x_n\}$, loss $\ell$, optimizer update function $update$, secrets $S$, secret protection parameters $(p_j, r_j)$, number of rounds $T$, target batch size $B$, clip norm $C$, noise multiplier $\sigma$, initial model $\theta_0$.
\begin{algorithmic}[1]
\State For each secret $j$, $\mu_j \leftarrow KL(Bern(r_j), Bern(p_j))$.
\State $\{w_i\} \leftarrow$ solution to \eqref{eq:secret_lp} using the chosen $\mu_j$ values.
\State For each example $i$, $\rho_i = \frac{B w_i}{\sum_{i'} w_{i'}}$
\State For each secret $j$, $\sigma_j = \min \{\sigma: KL(N(\sum_{x_i \in D_j} Bern(\rho_i), \sigma^2)^T, N(0, \sigma^2)^T) \leq \mu_j\}$
\State $\sigma = \max_j \sigma_j$
\State $\clip{\bfv}{C} := \bfv \cdot \min\{1, C / \ltwo{\bfv}\}$
\For{$t \in [T]$}
\State $\calB_i \leftarrow$ batch including $x_i$ independently w.p. $\rho_i$.
\State $\bfg_i \leftarrow \frac{1}{B}\sum_{x_i \in \calB_i} \nabla \clip{\ell(\theta_{i-1}; x_i)}{C}$
\State $\bfz_i \leftarrow N(0, \frac{C^2 \sigma^2}{B^2} \mathbb{I})$
\State $\theta_i \leftarrow update(\theta_{i-i}, \bfg_i + \bfz_i)$. \Comment{e.g. for SGD, $update(\theta, \bfg) = \theta - \eta \bfg$}
\EndFor
\end{algorithmic}
\end{algorithm}
\end{figure}

\section{Experiments}\label{sec:experiments}

\subsection{Secret Protection}\label{sec:expsecretprotection}

We now empirically investigate the effectiveness of our algorithm for training tailored to \cref{def:secretprotection}. Again, since we wish to use a public dataset (without actual business secrets), we use the arXiv dataset \cite{clement2019arxiv} as a benchmark, as it still represents a collection of knowledge on specialized topics. The arXiv dataset consists of 1.9 million papers uploaded to arXiv.org, including their abstract. We randomly split the dataset into a training set of 1.7 million examples, a test set of 0.1 million examples, and a validation set of 0.1 million examples. We fine-tune a pretrained BERT-Tiny model on the abstracts of examples in the arXiv dataset for the masked language modeling task as defined in \cite{devlin2019bert}. We sort all words appearing in any abstract by the number of abstracts they appear in, take the 50001st to 150000th words in decreasing order, and treat these as secrets. We treat each abstract as containing a secret if it contains the corresponding word. This set of secrets each appears roughly 50 to 100 times each. We use a fixed $p_j = 10^{-10}$ and assign each secret a uniformly random $r_j \in [2 \cdot 10^{-4}, 10^{-3}]$. We discard any examples not containing secrets for simplicity. We fix a batch size of $2048$ and $2000$ iterations of training, and sweep the learning rate. We fix the clip norm to be $1.0$ as we found this to be reasonably effective in most settings. To simulate training with more resources on a larger dataset, in our privacy accounting we assume our batch size and dataset size are 10 times larger, which is equivalent to reducing the noise multiplier by a factor of 10\footnote{Since we decrease the scale of the noise, this should only make any comparisons more unfavorable for our LP-based approach.}.

We compare our linear program-based preprocessing with the baseline of no preprocessing. For simplicity in \cref{def:secretprotection} we let the predicates $T_j$ be always 1, i.e. two adjacent datasets can differ in the examples containing the sensitive secret $s_j$ arbitrarily. In Figure~\ref{fig:fraction-sp} we first show for each $c = 2^k$ where $k \in \{-6, -5, \ldots, 4\}$ what fraction of examples are retained. For setting the constant in the linear program to be $c = 2^4$, our LP solution assigns all examples a weight of $1$, i.e. does no preprocessing of the dataset. So we will compare these choices of $c$, and demonstrate improvements over $c = 2^4$. In Figure~\ref{fig:nonpriv_testloss-sp} we compare the test losses achieved by training without noise on the resulting datasets, where we see only a small $\approx 0.04$ increase in test loss due to the pre-processing step even though it discards many examples. In contrast, in Figure~\ref{fig:nm-sp}, we compare the noise multipliers we need to achieve the desired bounds on posterior reconstruction probabilities via KL divergence bounds. We see that our preprocessing can lead to a significant multiplicative reduction of $\approx 8$ in the noise multiplier compared to no preprocessing. In Figure~\ref{fig:testloss-sp} we plot the test losses and see that the noise multiplier reduction from preprocessing strongly outweighs the small benefits from having the whole dataset, leading to a substantial decrease in the loss. 

\begin{figure}[H]
\centering
\begin{minipage}{0.4\textwidth}
    \centering
    \includegraphics[width=0.95\textwidth]{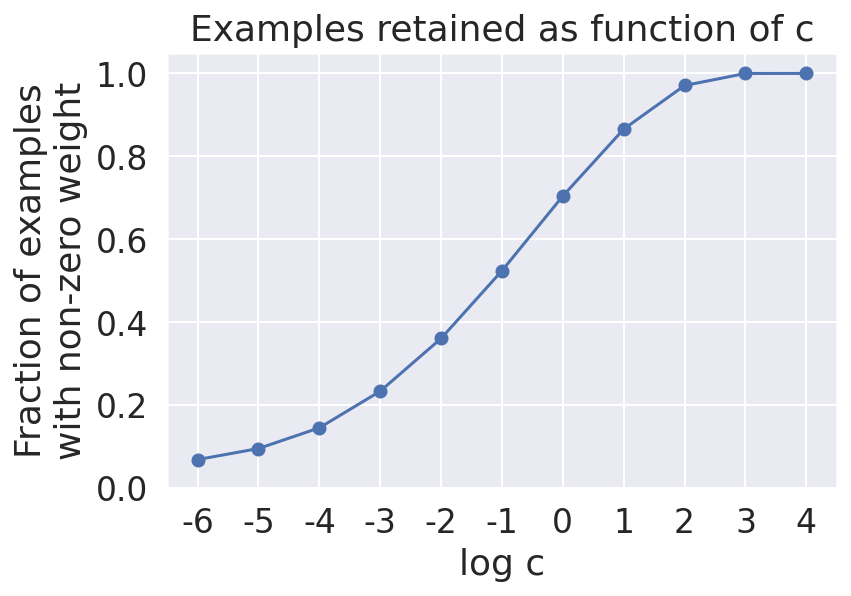}
    \subcaption{The fraction of examples in the dataset which are assigned non-zero weight.}
    \label{fig:fraction-sp}
\end{minipage}%
\hspace{1cm}
\begin{minipage}{0.4\textwidth}
    \centering
    \includegraphics[width=0.95\textwidth]{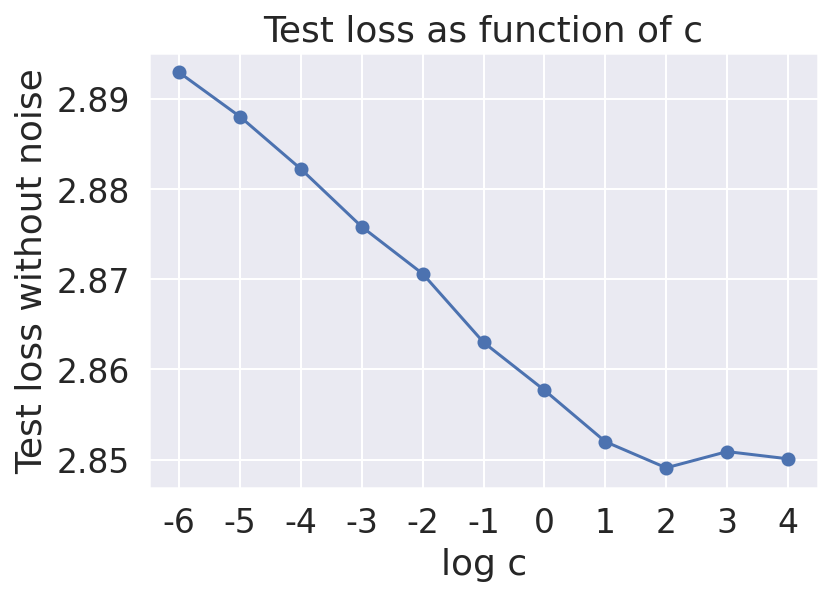}
    \subcaption{The test loss achieved by the different LP solutions, without any noise.}
    \label{fig:nonpriv_testloss-sp}
\end{minipage}
\begin{minipage}{0.4\textwidth}
    \centering
    \includegraphics[width=0.95\textwidth]{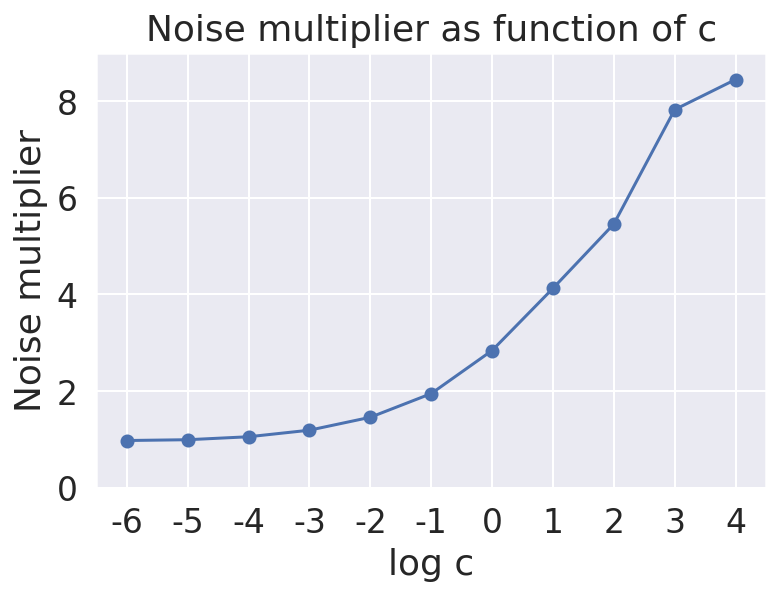}
    \subcaption{The noise multiplier needed for the desired reconstruction probability bounds.}
    \label{fig:nm-sp}
\end{minipage}%
\hspace{1cm}
\begin{minipage}{0.4\textwidth}
    \centering
    \includegraphics[width=0.95\textwidth]{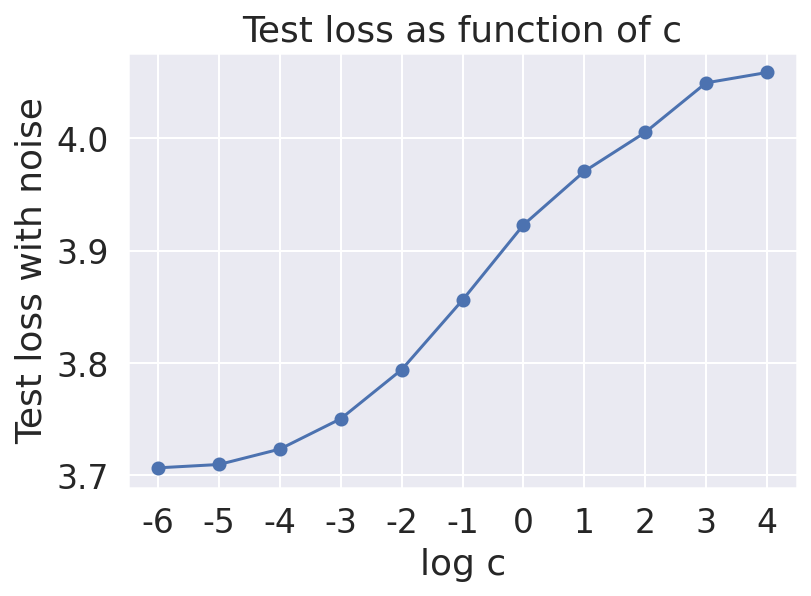}
    \subcaption{The test loss achieved by the different LP solutions.}
    \label{fig:testloss-sp}
\end{minipage}
\caption{Comparison of LP solutions for different values of $c$. In  \cref{fig:testloss-sp}, the 95\% confidence intervals have radius at most $.001$.}
\end{figure}

\mypar{Infeasibility of the Blow-Up Function} In our experiments we attempted to use the privacy loss distribution to compute the blow-up function, which would allow us to use the tight bounds on reconstruction probability of \cite{hayes2023bounding}. Using the \texttt{dp\_accounting} tools for computing the privacy loss distribution, we get a discrete privacy loss distribution with support $v_1, \ldots, v_k$ and probabilities $p_1, \ldots, p_k$. The corresponding blow-up function evaluated at $\sum_{i = j}^k p_i$ has a value of $\sum_{i = j}^k e^{v_i} p_i$, and linearly interpolates between these points. In particular, by definition of the blow-up function it should be the case $\sum_{i=1}^k e^{v_i} p_i = 1$. In our experiments, when computing the privacy loss distribution for two datasets differing in a single example at the noise multipliers in Figure~\ref{fig:nm-sp}, we did find the equality $\sum_{i=1}^k e^{v_i} p_i = 1$ held. However when we did the same calculation for two datasets differing in a secret ($\approx 100$ examples), we found instances where $\sum_{i=1}^k e^{v_i} p_i \geq 2 \cdot 10^4$. In other words, the numerical instability and pessimism in PLD accounting (specifically for \textit{Mixture of Gaussians mechanisms}, i.e. Gaussian mechanisms whose sensitivity's support is larger than $\{0, 1\}$) heavily inflates the blowup function. For this reason we opted to use the KL divergence-based reconstruction bounds of \cref{lem:generalized_fano} instead.

\subsection{Reconstructing Pre-Training data}

While our primary focus is on protecting pre-defined secrets, the underlying principle—that moderate indistinguishability guarantees, potentially weaker than standard DP, can suffice to limit reconstruction attacks can also apply to the broader problem of preventing training data extraction from large language models. Large language models (LLMs) are known to memorize portions of their training data, which can then be extracted \cite{nasr2023scalable,carlini2021extracting, carlini2022quantifying}. This section focuses on reconstruction attacks (also known as extraction attacks) against large language models to validate that moderate indistinguishability guarantees suffice to protect against reconstruction.  These attacks aim to extract training data from the model. We adopt a similar experimental setup and definition of memorization. 

\begin{definition}[(k-memorization )~\cite{carlini2022quantifying}]

A string $s$ is extractable with $k$ tokens of context from a model $f$ if there exists a (length-$k$) string $p$, such that the concatenation $[p || s]$ is contained in the training data for $f$, and $f(p) = s$ using greedy decoding.

\end{definition}

It's important to note that this definition might not encompass all possible forms of memorization~\cite{ippolito2022preventing}. However, it serves as a valuable proxy for measuring the extent of data extraction and provides a benchmark for comparison. This metric is widely used in various contexts, including copyright infringement cases, to assess the unauthorized reproduction of training data.

\mypar{Pretraining Setup} For our experiments, we pretrained a series of causal language models using the Gemma-2 architecture~\cite{gemmateam2024gemma2improvingopen} and the fineweb-edu dataset~\cite{penedo2024finewebdatasetsdecantingweb} tokenized using the Gemma-2 tokenizer. We explored various model sizes within the Gemma-2 family.

\mypar{Extraction Experiments}
To evaluate the reconstruction attack on our trained models, we fix a subset of the dataset. We prompted each model with the first $n$ tokens of a sequence from the subset and observed whether it could generate the subsequent $m$ tokens.  Our evaluation focused exclusively on greedy sampling, where the model selects the most probable token at each step.

\begin{figure}[H]
\centering
\begin{minipage}{0.49\textwidth}
    \centering
    \includegraphics[scale=0.5]{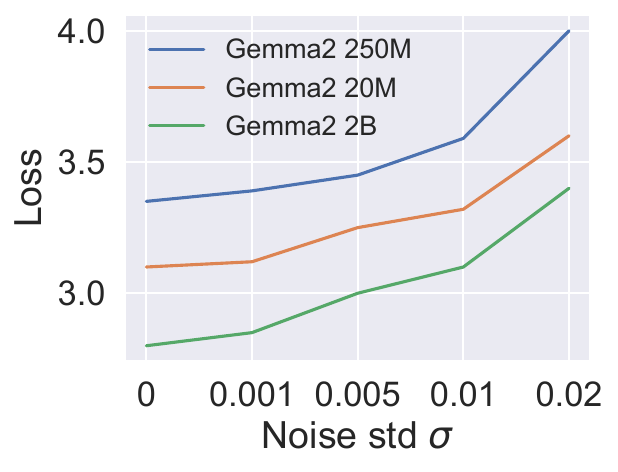}
    \caption{Effect of noise on the validation loss}
    \label{fig:loss_pretraing}
\end{minipage}
\hfill
\begin{minipage}{0.49\textwidth}
    \centering
    \includegraphics[scale=0.5]{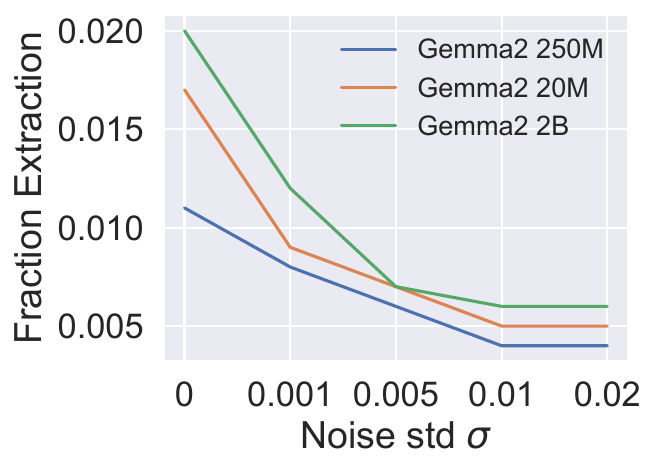}
    \caption{Fraction memorized}
    \label{fig:mem_frac}
\end{minipage}
\end{figure}

\mypar{Results} Figure~\ref{fig:loss_pretraing} demonstrates the impact noise addition has on training. With small noise multiplier, the effect on model utility is minimal. Figure~\ref{fig:mem_frac} highlights the effectiveness of small noise multipliers in mitigating memorization.  A small amount of noise leads to a substantial reduction in memorized instances. Interestingly, we see even with large noise multiplier ($\geq 0.01$) there is some remaining memorization. 
We manually looked at the memorized examples for higher noise multipliers and observed that remaining instances are primarily attributed to generalization. For example, prompting the model with "12345" results in the continuation "6789", even with increased noise levels. This suggests the model learns underlying patterns rather than simply memorizing specific sequences.

\bibliographystyle{alpha}
\bibliography{ref} 
\end{document}